\documentclass{ws-ijfcs-mod}
\usepackage{enumerate}
\usepackage{url}
\usepackage{hyperref}
\usepackage[capitalize]{cleveref}
\usepackage{todonotes}

\usepackage[utf8]{inputenc}
\usepackage[english]{babel}

\usepackage{tikz}
\usetikzlibrary{trees}

\newcommand{\A}{\mathcal{A}}
\newcommand{\B}{\mathcal{B}}

\newcommand{\F}{\mathcal{F}}
\newcommand{\T}{\mathcal{T}}

\newcommand{\nats}{\mathbb{N}}

\newcommand{\trees}[1][\Sigma]{\mathcal{T}_{#1}}
\newcommand{\dom}{\mathit{dom}}

\newcommand{\pot}{\mathit{paths}} 
\newcommand{\tfp}{\mathit{trees}} 

\newcommand{\phipath}{\varphi_\mathrm{pth}} 
\newcommand{\phitree}{\varphi_\mathrm{tr}} 
\newcommand{\phiT}{\varphi_T} 
\newcommand{\phinT}{\varphi_{\neg T}}

\begin{document}

\markboth{Thomas, L\"oding}
{On  the Boolean Closure of DTDA}

%
\catchline{}{}{}{}{}
%

\title{On  the Boolean Closure of Deterministic Top-Down Tree Automata}

\author{Christof L\"oding, Wolfgang Thomas}

\address{RWTH Aachen University, Informatik 7, Aachen, Germany\\
\email{$\{$loeding,thomas@automata.rwth-aachen.de$\}$}
}
%
%

\maketitle


\begin{abstract}
The class of Boolean combinations of tree languages recognized by deterministic 
top-down tree automata (also known as deterministic root-to-frontier automata) 
is studied. The problem of determining for a given regular tree language whether 
it belongs to this class is open. We provide some progress by two results: 
First, a characterization of this class by a natural extension of 
deterministic top-down tree automata is presented, and as an application we 
obtain a convenient method to show that certain regular tree languages 
are outside this class. In the second result, it is shown that, for fixed $k$,  it is decidable 
whether a regular tree language is a Boolean combination of $k$ tree languages 
recognized by deterministic top-down tree automata.
\end{abstract}


\section{Introduction}	

In the theory of finite automata over finite trees, the principal model is the ``bottom-up"
(or ``frontier-to-root'') tree automaton; in either its deterministic 
or nondeterministic version it serves to characterize the ``regular tree languages''. 
It is well-known that the nondeterministic ``top-down" (or``root-to-frontier") tree automata are 
of the same expressive power whereas the deterministic ones are strictly 
weaker. In the sequel we shall call the latter deterministic top-down tree automata
(DTDA) and denote the class of tree languages recognized by them ${\cal T}(\rm{DTDA})$.

The expressive power of DTDA is  well understood. For this purpose one extracts 
from a tree $t$ the set $\pot(t)$ of finite paths in $t$ from root to frontier (in a natural way 
described in detail below) and thus gets, for a tree language $T$, the set 
$\pot(T) = \bigcup \{\pot(t) \mid t \in T\}$. Then a regular tree language $T$ is recognizable 
by a DTDA iff $T$ is fixed by $\pot(T)$ in the sense that
\begin{equation*}
  t \in T \ \ \mbox{iff}  \ \ \pot(t) \subseteq \pot(T).
  \tag{$*$}
\end{equation*}
A closer analysis of $(*)$ shows that recognizability of a regular tree language by a DTDA is 
decidable (cf. Vir\`{a}gh \cite{Viragh80}).  The proof of this result also appears in the first 
(and up to this time) only monograph on tree automata \cite{GecsegS84} by 
Gecseg and Steinby. As one of the founders of the theory of tree automata, Magnus Steinby has contributed numerous substantial results to this subject, among them several results on deterministic top-down tree automata (see, e.g. \cite{GecsegS78} and \cite{JurvanenS2019}). It is our pleasure and honor to devote this paper to the memory of Magnus, a great colleague and friend.

DTDA are very limited in expressive power. The class $\cal{T}({\rm DTDA})$ 
is neither closed under union nor under complement. This is apparent 
already from a very simple example  tree language that is not DTDA-recognizable, 
namely the  two element set $T_0 = \{f(a,b), f(b,a)\}$ (using the term notation for trees). 
If $T_0$ was DTDA-recognizable, then invoking $(*)$ also the trees $f(a,a)$ and $f(b,b)$ would belong
to $T_0$, a contradiction. 

In order to overcome this weakness of DTDA, several extensions of the class of 
DTDA-recognizable tree languages (within the class of regular tree languages) 
have been considered, among them the closure of 
${\cal T}({\rm DTDA})$ under Boolean operations \cite{Jurvanen92} and the class of tree languages 
recognized by DTDA with ``regular frontier-check'' \cite{JurvanenPT93}. In the present paper we study 
the Boolean closure of ${\cal T}({\rm DTDA})$, denoted ${\rm Bool}({\cal T}({\rm DTDA}))$, 
extending the work of Jurvanen \cite{Jurvanen92}. A major problem, open now for several decades,
 is to show the analogue of Vir\'{a}gh's result mentioned above, i.e., to show that membership 
 of a regular tree language in ${\rm Bool}({\cal T}({\rm DTDA}))$ is decidable (or to show 
 the unlikely opposite -- that it is undecidable). 
 
 While leaving open this question, the present paper offers two results that provide some progress. 
 First we present a natural model of automaton (``DTDA with set acceptance") 
 that characterizes ${\rm Bool}({\cal T}({\rm DTDA}))$
 and permits simple proofs that certain regular tree languages do not belong to this class. 
 Second, we provide a decision procedure for subclasses of ${\rm Bool}({\cal T}({\rm DTDA}))$, 
 where only a bounded number of Boolean operations is admitted on top of  ${\cal T}({\rm DTDA})$. 
 More precisely, we consider, for $k \geq 1$, the class $k$-${\rm Bool}({\cal T}({\rm DTDA}))$, 
 containing the tree languages that are Boolean combinations of $k$ tree languages from 
 ${\cal T}({\rm DTDA})$. We show that for each $k \geq 1$, membership of a regular tree 
 language in the class $k$-${\rm Bool}({\cal T}({\rm DTDA}))$ is decidable.

 Let us mention three papers on related work.

In \cite{Jurvanen92}, Eija Jurvanen provided the first study of the class
${\rm Bool}({\cal T}({\rm DTDA}))$, showing algebraic properties and
giving first examples of regular tree languages outside this class.
Using the DTDA with set acceptance of the present paper we get such
examples without entering complex combinatorial calculations as they
appear in \cite{Jurvanen92}.

In his PhD thesis \cite{BojPhD} of Mikolaj Boja{\`n}zcyk obtained
a characterization of ${\rm Bool}({\cal T}({\rm DTDA}))$ in terms of
{\em word sum automata}, which are deterministic bottom-up tree automata
derived from word automata for paths. Our model of DTDA with
set acceptance has the advantage of being a direct generalization
of DTDA in which proofs of non-membership in ${\rm Bool}({\cal T}({\rm DTDA}))$
are straightforward.

Finally, we mention the recent paper \cite{LeupoldM21} of P.~Leupold and S.~Maneth
in which -- among other results -- decidability of a subclass of
${\rm Bool}({\cal T}({\rm DTDA}))$ is claimed, namely of the class of finite
unions of tree languages in ${\cal T}({\rm DTDA})$. Since Lemma~5 of
that paper is not valid (as a counter-example shows), this decidability
claim seems open to us (parts of the results from \cite{LeupoldM21} have been fixed in \cite{ManethS24}, but this does not include the decidability of the class of finite unions of tree languages in ${\cal T}({\rm DTDA})$.)

\section{Preliminaries}

Ranked trees are terms composed from symbols of a ranked alphabet, 
which is a finite set of symbols where each symbol has a fixed rank, 
a natural number that indicates the arity of this symbol. Thus 
symbols of rank $k > 0$ can be considered as function symbols of 
arity $k$ and symbols of rank 0 as constants. We denote a ranked 
alphabet by the letter $\Sigma$, the rank of symbol $a$ as $|a|$, and 
the set of symbols of $\Sigma$ of rank $i$ by $\Sigma_i$.  
Ranked trees over $\Sigma$ are also called $\Sigma$-trees. The set 
$\trees$ of $\Sigma$-trees is defined inductively by the two 
clauses that each symbol $c \in \Sigma_0$ belongs to $\trees$ and 
that with $t_1, \ldots, t_i \in \trees$ and $f \in \Sigma_i$ also
$f(t_1, \ldots, t_i)$ belongs to $\trees$.  

It is convenient to use also a graphical representation (as the 
word ``tree'', rather than ``term'', suggests). Here a tree consists 
of nodes that are labelled with letters from 
$\Sigma$. Formally, for a tree $t$ we introduce its domain 
$\dom(t) \subseteq \nats^*$\index{$d$@$\dom(t)$}. If the 
tree is a constant then this domain just contains 
$\varepsilon$; if $t = a(t_1, \ldots, t_k)$ with $k > 0$ then 
$\dom(t) = \{\varepsilon\} \cup \bigcup_{i=1}^{k} i \cdot \dom(t_i)$, 
where $\cdot$ is the concatenation of words. We say that $\varepsilon$ 
is the root, that a node $u$ with a symbol of rank $k > 0$ has $u1, \ldots, uk$ 
as its children, that nodes without children are leaves, and that 
non-leaf-nodes are inner nodes. A tree $t $ over $\Sigma$ can then be 
presented as a map $t : \dom(t) \rightarrow \Sigma$. The nodes 
of $\dom(t)$ without successors in $\dom(t)$ form the ``frontier'' of 
$t$; this set is denoted ${\rm fr}(f)$. Sometimes we also use successors 
 of frontier nodes and call them nodes of the ``outer frontier''. We set 
${\rm fr}^+(t) := \{u 1 \mid  u \in {\rm fr}(t)\}$ and 
 $\dom^+(t) := \dom(t) \cup {\rm fr}^+(t)$.

Now we recall the standard concepts of deterministic tree automata in 
the two versions ``bottom-up" and ``top-down".

A deterministic {\em bottom-up 
tree automaton} over $\Sigma = \bigcup_{i = 0}^{r} \Sigma_i $
is of the form $\A = (Q,  \delta, F)$, where $Q$ is the finite set of states, 
$F \subseteq Q$ is the set of ``final'' (or ``accepting'') states, and $\delta$ the transition 
function, given as the union of functions $\delta_i : Q^i \times \Sigma_i \rightarrow Q$ 
for $i = 0, \ldots , r$. Here $\delta_0$ is considered as a function from $\Sigma_0$ to $Q$. 
We define, for any tree $t$ over $\Sigma$, the state  $\delta_\A(t)$ 
inductively, by setting, for $a \in \Sigma_0$, $\delta_\A(a) := \delta_0(a)$, and for 
$a \in \Sigma_i$ where $i \in \{1, \ldots, r\}$, $\delta_\A(a(t_1, \ldots, t_i)) := 
\delta_i(\delta_\A(t_1) , \ldots, \delta_\A(t_i), a)$. The tree $t$ is accepted by $\A$ if 
$\delta_\A(t) \in F$, and the tree language recognized by $\A$ is 
$T(\A) = \{t \in T_\Sigma \mid t \mbox{ is accepted by } \A\}$.  
The tree languages recognized by bottom-up tree automata are the ``regular'' ones. 

A deterministic {\em top-down  tree automaton} (DTDA) over the ranked alphabet 
$\Sigma = \bigcup_{i = 0}^{r} \Sigma_i $ is of the form 
$\A = (Q, q_0,  \delta, F)$ where $Q$ is the finite set of states, 
$q_0 \in Q$, $F \subseteq Q$, and $\delta$ again 
given as  the union of functions $\delta_i$ ($i \in \{0,\ldots, r\}$), now with 
$\delta_i: Q \times \Sigma_i \rightarrow Q^i$ for $i \in \{1, \ldots, r\} $ and 
$\delta_0 : Q \times \Sigma_0 \rightarrow Q$. We define the function 
$\delta_\A: Q \times T_\Sigma \rightarrow 2^Q$, 
associating with each tree $t$ the ``set of states reached by $\A$ from $q$ via $t$ 
at the outer frontier'', as 
follows: If $t = a \in \Sigma_0$ then $\delta_\A(q, t) = \{\delta_0(q, a)\}$. If $t = a(t_1, \ldots , t_i)$, 
say with $\delta_i(q, a) = (q_1, \ldots, q_i)$, then 
$\delta_\A(q, t) = \delta_\A(q_1, t_1) \cup \ldots \cup \delta_\A(q_i, t_i)$. The automaton $\A$ 
accepts $t$ if $\delta_\A(q_0, t) \subseteq F$. As already said  in the introduction, we denote 
by ${\cal T}({\rm DTDA})$ the class of tree languages recognized by DTDA. 

The behaviour of DTDA can best be understood in terms of the runs over the given 
input trees. Here we refer to the ``graphical'' presentation of trees and 
associate with any node $u$ of $\dom^+(t)$ the state reached at $u$ by the automaton under 
consideration. Contrary to the case of bottom-up automata, the run of a DTDA can be 
built up ``path-wise'': To determine the state at node $u$ it suffices to follow the path 
from the root to $u$ and apply the transition function of the automaton restricted to 
this path. Let us describe this well-known observation in precise terms.

A (labeled) path over a ranked alphabet $\Sigma$ is a finite word of the form $\pi = a_1d_1a_2d_2 \cdots a_n d_n a_{n+1}$ where
\begin{itemize}
\item $a_i \in \Sigma$ for all $i \in \{1, \ldots, n+1\}$, 
\item $d_i \in \{1, \ldots, |a_i|\}$ for all $i \in \{1, \ldots, n\}$, and
\item $a_{n+1} \in \Sigma_0$.
\end{itemize}
So $\pi$ alternates between labels and possible directions in the tree, ending with a leaf label. We denote the alphabet for paths by $\Gamma$, so
\[
\Gamma := \Sigma \cup D \text{ with $D := \{1, \ldots,r\}$}
\]
where $r$ is the maximal rank of a symbol in $\Sigma$. 

For a tree $t$ over $\Sigma$ we denote by $\pot(t) \subseteq \Gamma^*$ the set of all labeled paths in $t$. 
For example, let $\Sigma_2 = \{a,b\}$, $\Sigma_0 = \{c,d\}$, and $t$ be the left-most tree in Figure~\ref{fig:paths}, that is, $t = a(b(c,d),c)$. Then $\pot(t) = \{a1b1c, a1b2d, a2c\}$. We lift this to sets of trees by defining $\pot(T) := \bigcup_{t \in T} \pot(t)$. 

Vice versa, if we are given a set $P$ of labeled paths, then we define
\[
\tfp(P) := \{t \in \trees \mid \pot(t) \subseteq P\}
\]
to be the set of all trees that can be built from paths in $P$. An example is shown in Figure~\ref{fig:paths}.

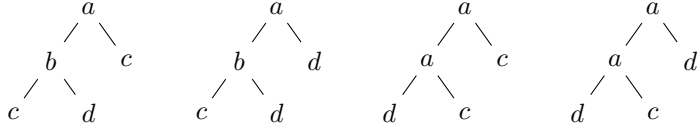
\begin{figure}
\begin{center}
\begin{tikzpicture}[grow via three points={%
one child at (0,-.7) and two children at (-.5,-.7) and (.5,-.7)}]
\node at (0,0) {$a$}
         child {node {$b$}
           child {node {$c$}} 
           child {node {$d$}}
         }
         child {node {$c$}}
         ;
\node at (2.5,0) {$a$}
         child {node {$b$}
           child {node {$c$}} 
           child {node {$d$}}
         }
         child {node {$d$}}
         ;
\node at (5,0) {$a$}
         child {node {$a$}
           child {node {$d$}} 
           child {node {$c$}}
         }
         child {node {$c$}}
         ;
\node at (7.5,0) {$a$}
         child {node {$a$}
           child {node {$d$}} 
           child {node {$c$}}
         }
         child {node {$d$}}
         ;
        
\end{tikzpicture}
\end{center}
\caption{The set $\tfp(P)$ for $P = \{a1b1c, a1b2d, a1a1d, a1a2c, a2c, a2d, b2c\}$. Note that the path $b2c$ is in $P$ but does not occur in any of the trees because no tree with $b$ at the root can be built from  paths in $P$.} \label{fig:paths}
\end{figure}


The following theorem characterizes DTDA-recognizable languages in terms of paths.
\begin{theorem}[\cite{Viragh80}]\label{thm:DTDApaths}
A tree language $T$ over $\Sigma$ is DTDA-recognizable iff \/ $T=\tfp(P)$ for a regular set $P \subseteq \Gamma^*$ of paths. 
\end{theorem}

%
%
%
%
%
%
%
%

\section{Automata with Set Acceptance}

We recall that a DTDA $\A = (Q, q_0, \delta, F)$ accepts a tree $t$ if $\delta_\A(t) \subseteq F$, i.e.,  {\em all}  states reached 
in the run of $\A$ over $t$ at the outer frontier 
 belong  to $F$. A DTDA with set acceptance will make 
more detailed use of the set $R$ of states reached at the outer frontier of the tree under consideration. 
In a standard DTDA this set should satisfy $R \subseteq F$. We shall allow also a dual condition, 
requiring that {\em some} state of $F$ is reached, which can be phrased as $R \cap F \not= \emptyset$ 
with $R$ specified as above. More generally, we shall allow an acceptance requirement for any combination 
on the existence or non-existence of states of $Q$. In other words, we allow for any subset $F$ of $Q$ 
the requirement that $\delta_\A(t) = F$, i.e., that precisely the states of $F$ are reached at the outer frontier. Collecting those $F$ which 
are admitted in this sense for acceptance of a tree, we obtain a family $\F$ of state sets as a new kind of acceptance 
component of DTDA.  

In the remainder of this section we only consider binary alphabets of the form $\Sigma = \Sigma_0 \cup \Sigma_2$. We make this assumption for notational simplicity, the results hold for arbitrary ranked alphabets.

A {\em DTDA with set acceptance} over the ranked alphabet $\Sigma$ 
is of the form 
${\cal A} = (Q, q_0, \delta, {\cal F})$ 
where $q_0$ and $\delta$ are as for DTDA and where ${\cal F} \subseteq 2^Q$ is a family 
of state sets over $Q$. The automaton ${\cal A}$ accepts the tree $t$ if 
$\delta_\A(q_0, t) \in {\cal F}$. 

This type of acceptance condition is familiar from the theory of automata over infinite words
and called there Muller condition (introduced in \cite{Muller63}); an infinite 
run is declared accepting if the states occurring infinitely often in it form a set that belongs to ${\cal F}$. 

Let us mention some immediate properties and examples of DTDA with set acceptance. 
First, as expected, the standard DTDA can be considered as a special case:  
\begin{remark}\label{extension}
Each DTDA can be presented as a DTDA with set acceptance.
\end{remark}

Consider a DTDA ${\cal A} = (Q, q_0, \delta, F)$ and observe  
that the tree automaton ${\cal A}' = (Q, q_0, \delta, \{R \in 2^Q \mid  R \subseteq F\} )$ with set acceptance is 
equivalent to ${\cal A}$. 

We recall that a singleton tree language $\{t\}$ can be recognized by a DTDA with a singleton 
set $F = \{q_+\}$ as acceptance component: A corresponding DTDA has, besides 
the accepting state $q_+$ and the rejecting state $q_-$,  
one state for each symbol occurrence in $t$. We denote the different occurrences of symbol $a$ by 
$a^i$ for different numbers $i$ and denote by $q_{a^i}$ the corresponding states. 
If the top symbol of $t$ is the symbol occurrence $a^i$ we take $q_{a^i}$ as initial state $q_0$. 
If in $t$ the triple $a^i (b^j, c^k)$ of symbol occurrences appears (where $a \in \Sigma_2$),  
we set $\delta_2(q_{a^i}, a) = (q_{b^j}, q_{c^k})$, and for a symbol occurrence $a^i$ with $a \in \Sigma_0$ 
we set $\delta_0(q_{a^i}, a) = q_+$; in all other cases we take $(q_-, q_-)$, respectively $q_-$, as 
value. It is clear that this DTDA recognizes the set $\{t\}$.
Proceeding to DTDA with set acceptance, we use the same automaton, but taking as acceptance 
component the singleton family  $\{ \{q_+\} \}$. So we have the following remark:

~
\begin{remark}\label{sing}
A singleton tree language is recognized by 
\begin{itemize}
\item a DTDA with a singleton as acceptance component,  
\item a DTDA with set acceptance whose acceptance component is a singleton consisting of a singleton.
\end{itemize}
\end{remark}

Next we verify that DTDA with set acceptance allow for closure under Boolean operations: 

\begin{proposition}\label{bool}
The class of tree languages recognized by DTDA with set acceptance (over a fixed ranked alphabet)  is 
closed under Boolean operations. 
\end{proposition}

\begin{proof}
It suffices to consider complementation and union. The proof is straightforward; we 
only give the definition of the required automata.

Given the tree automaton 
${\cal A} = (Q, q_0, \delta, {\cal F})$ with set acceptance, we may take ${\cal A}' = (Q, q_0, \delta, 2^Q \setminus {\cal F})$
as the complement automaton of ${\cal A}$. Given two DTDA ${\cal A}_1 = (Q_1, q^1_{0}, \delta^1, {\cal F}^1)$ 
and ${\cal A}_2 = (Q_2, q^2_{0}, \delta^2, {\cal F}^2)$ with set acceptance (and where $\delta^i = \delta^i_0 \cup \delta^i_2$), 
the following 
DTDA with set acceptance recognizes $T({\cal A}_1) \cup T({\cal A}_2)$: 
$$ {\cal A}' = (Q_1 \times Q_2, (q^1_0, q^2_0), \delta',  {\cal F}')$$
where $\delta' = \delta'_0 \cup \delta'_2$ with the definitions 
$\delta'_0((p,q), a)  = (\delta^1_0 (p,a), \delta^2_0(q,a))$  and,
using the projections $\pi_1, \pi_2$ with $\pi_1(p,q) = p$ 
and $\pi_2(p,q) = q$,  
$$
\delta'_2((p,q), a) = ((\pi_1(\delta^1_2 (p,a)), \pi_1(\delta^2_2 (q,a))), (\pi_2(\delta^1_2 (p,a)), \pi_2(\delta^2_2 (q,a))))
$$
Finally, we define ${\cal F}'$ by specifying that 
 $R = \{(p_1, q_1), \ldots, (p_k, q_k)\}$ belongs to ${\cal F}'$ iff 
$\{p_1, \ldots, p_k\} \in {\cal F}^1$ or $\{q_1, \ldots, q_k\} \in {\cal F}^2$. 
\end{proof}

As a consequence of Remark \ref{sing} and Proposition~\ref{bool}, we see that each finite tree language $\{t_1, \ldots, t_n\}$ 
(and hence also each co-finite tree language) is 
recognizable by a DTDA with set acceptance. As mentioned above, standard DTDA fail to recognize 
the two-element tree language $\{f(a,b), f(b,a)\}$; so we have shown that the expressive power 
of DTDA   with set acceptance properly exceeds that of DTDA. 

Next we verify that DTDA with set acceptance characterize the 
Boolean closure of ${\cal T}({\rm DTDA})$. Intuitively speaking, we see that adding Boolean 
logic on the occurrence of states to the acceptance condition of DTDA (i.e.,  ``inside'' DTDA) 
amounts to the same as adding Boolean logic to the tree languages recognized by DTDA (i.e., ``outside'' DTDA). 

\begin{proposition}\label{charact}
A tree language $T$ is recognized by a DTDA with set acceptance iff T belongs to ${\rm Bool}({\cal T}({\rm DTDA}))$.
\end{proposition} 

\begin{proof}
The direction from right to left is clear from Remark \ref{extension} and Proposition~\ref{bool}. 

For the 
direction from left to right, consider a tree automaton ${\cal A} = (Q, q_0, \delta, {\cal F})$ with set acceptance. 
Let ${\cal F} = \{F_1, \ldots, F_k\}$. Then the tree language recognized by ${\cal A}$ is the union of 
the tree languages recognized by the DTDA ${\cal A}_i = (Q, q_0, \delta, \{ F_i \})$. So it suffices to show 
that a DTDA of the simple form ${\cal B} = (Q, q_0, \delta, \{F\})$ with set acceptance recognizes a tree language in 
${\rm Bool}({\cal T}({\rm DTDA}))$. Now ${\cal B}$ accepts a tree $t$ if 
\begin{itemize}
\item
each state reached via some path in the unique run of ${\cal B}$ 
on $t$ belongs to $F$, and that conversely 
\item each $F$-state is reached via some path in the unique run of ${\cal B}$ on $t$. 
\end{itemize}
So $T({\cal B})$ is the intersection of the tree language recognized by the DTDA $(Q, q_0, \delta, F)$ and, 
writing $F = \{q_1, \ldots, q_r\}$, of the intersection of the tree languages $T_{q_i}$ where 
$T_{q_i}$ contains the trees on which the given DTDA reaches 
$q_i$ somewhere at the outer frontier. 
Now $T_{q_i}$ belongs to ${\rm Bool}({\cal T}({\rm DTDA}))$ 
since the complement of $T_{q_i}$ is recognized by the standard DTDA 
$(Q, q_0, \delta, Q \setminus \{q_i\})$. 
\end{proof}
%

As an application of Proposition \ref{charact}, we present examples of regular tree languages that 
are outside ${\rm Bool}({\cal T}({\rm DTDA}))$.

DTDA with set acceptance can check trees for the existence or the non-existence of 
paths (with regular properties) but this does not cover conditions on the {\em number} of occurrences 
of paths or on the (left-to-right) {\em order} of different paths. As we shall show, these properties of 
(or relations between) paths take us outside ${\rm Bool}({\cal T}({\rm DTDA}))$. Regarding the aspect
of ordering, we say that symbol $a$ occurs left to symbol $b$ in a tree if $a$ and $b$ 
are the labels of two nodes $u$ and  $v$ 
such that for the last common ancestor $w$ of them $u$ belongs to the left subtree of $w$ 
and $v$ to the right subtree of the node $w$. 

\begin{proposition}\label{twoexamples}
Let  $\Sigma = \Sigma_2 \cup \Sigma_0$ with $\Sigma_2 = \{a\}$ and $\Sigma_0 = \{c, d, e\}$ and 
\begin{itemize}
\item 
$T_1 = \{ t \in \T_\Sigma \mid  d \mbox{ occurs only once in  } t\}$,  
\item
$T_2 = \{ t \in \T_\Sigma \mid  d \mbox{ occurs left of } \ e \mbox{ in  } t\}$. 
\end{itemize}
Neither $T_1$ nor $T_2$ belong to ${\rm Bool}({\cal T}({\rm DTDA}))$.
\end{proposition}
\begin{proof}
In both cases we proceed in the same way: Given a DTDA with set acceptance $\A_1$  recognizing $T_1$ 
we present a tree $t_1 \in T_1$ such that a modification of it is outside $T_1$ but still accepted by $\A_1$; similarly for a DTDA with set acceptance $\A_2$ recognizing $T_2$ 
and a suitable tree $t_2 \in T_2$.

We start with a ``right-comb'' where to the root symbol $a$ we attach $n$ times $a$ as right child, 
always taking $c$ as the left child, and at the last $a$ finishing with $c$ as left and right child. 
When increasing $n$,  the automaton $\A_1$ assumes an ultimately periodic sequence of states on 
the branch consisting of the letters $a$, say finishing the initial part of this state sequence after the $k$-th letter $a$, 
the first period after the $(k+p)$-th letter $a$, and the second period after the $(k+2p)$-th letter $a$.  

Let $t_1$ be the comb-tree with $(k+3p)$ letters $a$, where the left child of the $(k+1)$-th $a$ is changed from 
$c$ to $d$. Then $t_1 \in T_1$ and thus accepted by $\A_1$. 
If we change also left child of the $(k+p+1)$-th letter $a$ from $c$ to $d$, 
thus obtaining the tree $t'_1$, then $t'_1 \not\in T_1$ but 
by choice of $k$ and $p$ the same states will occur at the outer 
frontiers of $t_1$ and $t'_1$ in the runs of $\A_1$ on  these trees;  
so $\A_1$ accepts $t'_1$. See Figure~\ref{fig:example-one-non-bool-dtda} for an illustration.

\begin{figure}
\begin{center}
\begin{tikzpicture}[thick]

  \path (0,0)   node (root)    {$a$}    -- +(0,.3) node[right]{$q_0$}
  -- +(-.8,-.5) node (0)       {$c$}    
  -- ++(.5,-.5) node (1)       {$a$}    -- +(0,.3) node[right]{$q_1$}
  -- +(-.8,-.5) node (10)      {$c$}    
  -- ++(.8,-.8) node (k+1)     {$a$}    -- +(0,.3) node[right]{$q_k=q$}
  -- +(-.8,-.5) node (k+10)    {$d$}    
  -- ++(.8,-.8) node (k+p+1)   {$a$}    -- +(0,.3) node[right]{$q_{k+p}=q$}
  -- +(-.8,-.5) node (k+p+10)  {$c/d$}  
  -- ++(.8,-.8) node (k+2p+1)  {$a$}    -- +(0,.3) node[right]{$q_{k+2p}=q$}
  -- +(-.8,-.5) node (k+2p+10) {$c$}    
  -- ++(.8,-.8) node (k+3p)    {$a$}    -- +(0,.3) node[right]{$q_{k+3p}$}
  -- +(-.8,-.5) node (k+3p0)   {$c$}    
  -- ++(.5,-.5) node (k+3p1)   {$c$}    
  ;

  \draw[-] (root) edge (0) edge (1)
  (1) edge (10) edge[dotted] (k+1)
  (k+1) edge (k+10) edge[dotted] (k+p+1)
  (k+p+1) edge (k+p+10) edge[dotted] (k+2p+1)
  (k+2p+1) edge (k+2p+10) edge[dotted] (k+3p)
  (k+3p) edge (k+3p0) edge (k+3p1)
  ;

\end{tikzpicture}
\caption{Illustration for the proof of Proposition~\ref{twoexamples}. All inner nodes of the tree are labeled with~$a$. All leaves not shown in the illustration are labeled with $c$.} \label{fig:example-one-non-bool-dtda}
\end{center}
\end{figure}
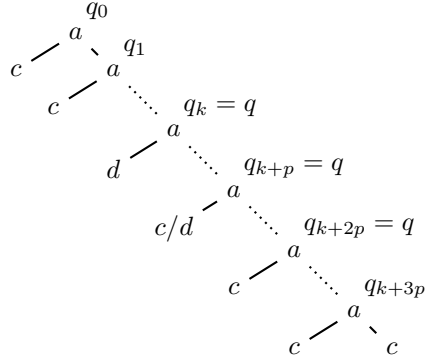
%
%

Given $\A_2$ we again take the comb-trees mentioned above and define $k$ and $p$ analogously. 
Let $t_2$ be the comb-tree with $(k+3p)$ letters $a$, where the left child of the $(k+1)$-th $a$ is changed from 
$c$ to $d$ and the left child of the $(k+p+1)$-th $a$ is changed from $c$ to $e$. Then $t_2 \in T_2$, thus accepted by $\A_2$. 
If we now change the left child of the $(k+1)$-th letter $a$ back to $c$, not touching the $e$ at the
left child of the $(k+p+1)$-th letter $a$, but changing the left child of the $(k+2p +1)$-th letter $a$ from $c$ to $d$, 
then we obtain a tree $t_2' $ outside $T_2$ but again by choice of $k$ and $p$ the same states will occur at the outer 
frontier of $t_2$ and $t'_2$ in the run of $\A_2$ on these trees, and thus $\A_2$ accepts $t_2'$. 

\end{proof}

We add a remark regarding the DTDA with regular frontier-check from \cite{JurvanenPT93}, mentioned
in Section 1. For the acceptance of input trees, such a DTDA $\A$ (say with state set $Q$)
is equipped with a standard finite automaton $\B$ on words whose input alphabet is $Q$.
A tree $t$ is accepted by $\A$ if the sequence of states reached by $\A$ on the outer frontier
$fr^+(t)$, read from left to right, is accepted by $\B$. The DTDA with set acceptance
correspond to the special case where the automata $\B$ only record the sets
of states from $Q$ that occur in given sequences over $Q$. The two tree languages
$T_1, T_2$ of Proposition 6 are easily seen to be recognizable by DTDA with
regular frontier-check, thus showing in a simpler way than in \cite{JurvanenPT93} that
${\rm Bool}({\cal T}({\rm DTDA}))$ is a proper subclass of the class of tree languages recognized by
DTDA with regular frontier-check. 

\section{Decidability of Membership in   $k$-${\rm Bool}({\cal T}({\rm DTDA}))$}    

In this section, we show that it is decidable whether a given regular tree language $T$ is a Boolean combination 
of at most $k$ deterministic top-down languages for a given $k$. We achieve this by invoking a decidability 
result on monadic second-order logic (MSO logic), namely that the MSO-theory of the infinite tree of given 
finite branching is decidable. We assume that the reader is acquainted with the basics of MSO logic (as 
given, e.g., in \cite{Thomas97}). 

We make use of the characterization of deterministic top-down languages in terms of their paths as given
 in Theorem~\ref{thm:DTDApaths}, which implies that we can encode DTDA-recognizable tree languages 
over $\Sigma$ by subsets of $\Gamma^*$ (recall that $\Gamma = \Sigma \cup D$ with $D = \{1, \ldots r\}$
where $r$ is the maximal rank of symbols in $\Sigma$). The following 
 characterization of membership in $k$-${\rm Bool}({\cal T}({\rm DTDA}))$ 
 provides the bridge to MSO logic, more precisely to MSO logic over the infinite $|\Gamma|$-branching tree:

\begin{lemma} \label{lem:k-bool}
  A tree language $T$ is in $k$-${\rm Bool}({\cal T}({\rm DTDA}))$ 
   iff there are regular sets $P_1, \ldots, P_k \subseteq \Gamma^*$ such that for all trees $t,t' \in \trees$ 
   with $t \in T$ and $t' \notin T$, there exists $i \in \{1, \ldots, k\}$ such that $(\pot(t) \subseteq P_i) \Leftrightarrow (\pot(t') \not\subseteq P_i)$. 
\end{lemma}
\begin{proof}
  For the left-to-right direction, assume that $T$ is in $k$-${\rm Bool}({\cal T}({\rm DTDA}))$. 
  Then $T$ is a Boolean combination of DTDA-recognizable languages $T_1, \ldots, T_k$.  Let $P_1, \ldots, P_k \subseteq \Gamma^*$ be regular path languages with $\tfp(P_i) = T_i$ according to Theorem~\ref{thm:DTDApaths}. 
  Note that for $t \in \trees$,  $\pot(t) \subseteq P_i$ is equivalent to $t \in T_i$.
  

  Write the Boolean combination of $T_1, \ldots, T_k$ for $T$ in disjunctive normal form, that is, as $T = \bigcup_i \bigcap_j R_{ij}$ where each $R_{ij}$ is one of $T_1, \ldots, T_k$ or a complement of one of $T_1, \ldots, T_k$.

  Now let $t, t' \in \trees$ with $t \in T$ and $t' \notin T$. Since $t \in T$, there is an index $i_0$ such that $t \in  \bigcap_j R_{i_0 j}$. And since $t' \notin T$, there is 
  (for all $i$, hence for $i_0$) an index $h$ such that $t' \notin R_{i_0 h}$. If $R_{i_0 h} = T_\ell$, then $\pot(t) \subseteq P_\ell$ and $\pot(t') \not\subseteq P_\ell$. If $R_{i_0h}$ is the complement of $T_\ell$, then $\pot(t) \not\subseteq P_\ell$ and $\pot(t') \subseteq P_\ell$.

  For the right-to-left direction, assume that $P_1, \ldots, P_k \subseteq \Gamma^*$  exist as in the claim of the lemma. Let $T_i$ be the set of trees $t$ with 
  $\pot(t) \subseteq P_i$ and choose   
  $R_1, \ldots, R_k$ such that each $R_i$ is either $T_i$ or the complement of $T_i$. Then for all trees $t,t' \in \trees$ with $t \in T$ and $t' \notin T$, there exists $i \in \{1, \ldots, k\}$ such that $t \in R_{i}$ iff $t' \notin R_{i}$. Hence, $R_1 \cap \cdots \cap R_k$ is either a subset of $T$ or has empty intersection with $T$. So we can write $T$ as the union of those $R_1 \cap \cdots \cap R_k$ that are subset of $T$.
\end{proof}{
The condition of Lemma~\ref{lem:k-bool} can be expressed in MSO logic over the infinite $\Gamma$-tree as we explain in the following.

The (infinite) $\Gamma$-tree is the structure with domain $\Gamma^*$ and binary successor relations $S_\gamma$ for each $\gamma \in \Gamma$ with $S_\gamma(u,v)$ if $v=u\gamma$. We also say that $v$ is the $\gamma$-successor of $u$. The root of the tree is $\varepsilon$. In the context of the infinite $\Gamma$-tree, we refer to the elements of $\Gamma^*$ as nodes of the tree. 

We consider monadic second-order logic (MSO logic) over this structure, which is first-order logic extended with quantification over sets of nodes, and atomic formulas $x \in X$ for first-order variables $x$ and set variables $X$. In general, set variables are denoted by capital letters, and first-order variables by lower case letters.

As an example, consider the following formula with two free first-order variables $x,y$ that expresses that $x$ is a prefix of $y$. The formula uses a set quantifier for the variable $X$.  We use $S(z,z')$ as an abbreviation for $\bigvee_{\gamma \in \Gamma} S_\gamma(z,z')$, which means that $z'$ is some successor of $z$.
\[
\forall X\;(y \in X \land \forall z,z'\;(S(z,z') \land z' \in X \rightarrow z \in X)) \rightarrow x \in X
\]
This formula states that all sets that contain $y$ and are closed under predecessors also contain $x$, which means that $x$ is a prefix of $y$. 
For more details on MSO logic over infinite trees, in particular the following theorem, see \cite{Thomas97}:
\vspace{0.2cm}

\begin{theorem}[\cite{Rabin69}]
  \label{thm:Rabin}
  It is decidable whether a given MSO formula $\varphi(X_1,\ldots,X_k)$ over the $\Gamma$-tree is satisfiable. Furthermore, if the formula is satisfiable, one can construct regular subsets $P_1, \ldots, P_k$ of $\Gamma^*$ that make $\varphi$ true if each $X_i$ is interpreted as $P_i$. 
\end{theorem}

We now explain how to express the condition from Lemma~\ref{lem:k-bool} in MSO (without giving all the details). 
Recall that labeled paths are elements of $\Gamma^*$, so they correspond to elements of the $\Gamma$-tree. One can easily write a formula $\phipath(x)$ that is true iff the node $x$ corresponds to a labeled path. The formula just needs to express that on the prefixes of $x$, the successors alternate between $\Sigma$ and $D$, starting with $\Sigma$, ending with $\Sigma_0$, and the $D$-successors match with the arity of the previous $\Sigma$-successor.

For every tree $t \in \trees$, we have that $\pot(t) \subseteq \Gamma^*$. We say that a set $U \subseteq \Gamma^*$ of nodes {\em encodes} a tree if $U = \pot(t)$ for some tree $t$. Using $\phipath(x)$, it is not very difficult to write a formula $\phitree(X)$ that expresses that the set $X$ encodes a tree.
Such a formula needs to ensure that $X$ is non-empty, only contains elements that correspond to labeled paths, and the prefix closure $X'$ of $X$ has the following properties: (1) if $x \in X'$ ends in a letter $a \in \Sigma_i$, then $xj \in X'$ iff $j \in \{1,\ldots,i\}$, and (2) if $x \in X'$ ends in a direction $j \in D$, then there is a unique $a \in \Sigma$ with $x'a \in X'$.

Furthermore, given a regular language $T \subseteq \trees$ of trees, there is a formula $\phiT(X)$ that expresses that $X$ encodes a tree that is in $T$. Such a formula uses $\phitree(X)$, and furthermore expresses the existence of an accepting run of an automaton for $T$ on the tree encoded by $X$. The standard translations from finite automata to MSO (see \cite{Thomas97}) can easily be adapted to our encoding of finite tree as subsets of the $\Gamma$-tree. Similarly, there is a formula $\phinT(X)$ that expresses that $X$ encodes a tree that is not in $T$.

Now the condition from Lemma~\ref{lem:k-bool} can be expressed in MSO by the following MSO sentence.
\[
\exists X_1, \ldots, X_k \, \forall Y,Y' \left(\phiT(Y) \land \phinT(Y') \rightarrow \bigvee_{i=1}^k(Y \subseteq X_i \leftrightarrow Y' \not\subseteq X_i)\right)
\]
From Theorem~\ref{thm:Rabin}, we can conclude that if this sentence is true, then there are regular sets $P_1, \ldots, P_k$ for the $X_i$ that make the sentence true. Hence,
by Lemma~\ref{lem:k-bool}, the sentence is true iff $T$ is in $k$-${\rm Bool}({\cal T}({\rm DTDA}))$.
By Theorem~\ref{thm:Rabin}, membership of $T$ in $k$-${\rm Bool}({\cal T}({\rm DTDA}))$ is decidable:

\begin{theorem}\label{boundk}
It is decidable for a given regular tree language $T$ and a given number $k$ whether $T$ is a Boolean combination of (at most) $k$ DTDA-recognizable languages.
\end{theorem}

\section{Conclusion}

We have shown two results that provide some progress towards an effective characterization of 
the Boolean closure of the class of DTDA-recognizable tree languages. In the first result, 
we offered an automata-theoretic characterization; however, this characterization does 
not permit an effective decision whether a regular tree language belongs to ${\rm Bool}({\cal T}({\rm DTDA}))$. 
In the second result such an effective decision was provided; however, it refers to 
a fixed subclass $k$-${\rm Bool}({\cal T}({\rm DTDA}))$ of ${\rm Bool}({\cal T}({\rm DTDA}))$. 

We conclude with some remarks on open questions.

In Proposition~\ref{twoexamples}
 two phenomena are mentioned that cause a regular tree language to fall outside 
${\rm Bool}({\cal T}({\rm DTDA}))$, namely, multiple occurrence of (properties of) subtrees, 
and the order between occurrences of (properties of) subtrees. We may ask whether there are other 
conditions of this kind on regular tree languages that also induce non-membership in  ${\rm Bool}({\cal T}({\rm DTDA}))$. 

A related track of research is to search for a characterization of  ${\rm Bool}({\cal T}({\rm DTDA}))$ 
in terms of the transition structure of minimal bottom-up tree automata (which are determined up to 
isomorphism for a given regular tree language). Specifically, one may 
search for patterns in the transition function which induce the presence 
of counting multiplicities of subtrees or of distinguishing the order of subtrees.  
We recall that such ``forbidden patterns'' are well-known
in the classical theory of automata over finite words. For example, Sch\"utzenberger's Theorem on the 
characterization of star-free word languages says that a regular word language $L$ is star-free
iff the minimal deterministic automaton recognizing $L$ is ``permutation-free'', i.e., there is no word $w$ 
that induces a non-trivial permutation of states of the automaton (a property that can effectively be decided 
since $ |w| $ can be bounded in the number of states of the automaton).

The problem of deciding membership of regular tree languages in ${\rm Bool}({\cal T}({\rm DTDA}))$
can also be phrased in a simpler and direct manner by asking for certain bounds in the two results 
of the present paper: Given a regular 
tree language $T$ (presented, e.g., in terms of its minimal bottom-up tree automaton), can one provide a bound $K$ 
such that $T$ belongs to ${\rm Bool}({\cal T}({\rm DTDA}))$ iff a DTDA with set acceptance and $\leq K$ states 
recognizes $T$? 
Or, more on the track of Theorem \ref{boundk}, can one determine some $K$ such that $T$ belongs to 
${\rm Bool}({\cal T}({\rm DTDA}))$ iff $T$ belongs to $K$-${\rm Bool}({\cal T}({\rm DTDA}))$?

\bibliographystyle{ws-ijfcs}
\bibliography{../complete}




\end{document}